\documentclass{article}[12pt]

\usepackage{amsmath,amssymb,psfrag}
\usepackage{comment}
\usepackage{graphicx}
\newcommand{\ket}[1]{\ensuremath{|{#1}\rangle}}
\newcommand{\bra}[1]{\ensuremath{\langle{#1}|}}
\newcommand{\inner}[2]{\ensuremath{\langle{#1}|{#2}\rangle}}

\newcommand{\den}[1]{{\mathbb I_{#1}}}

\newcommand{\op}[1]{\mathsf{#1}}

\def\squareforqed{\hbox{\rlap{$\sqcap$}$\sqcup$}}
\def\qed{\ifmmode\squareforqed\else{\unskip\nobreak\hfil
\penalty50\hskip1em\null\nobreak\hfil\squareforqed
\parfillskip=0pt\finalhyphendemerits=0\endgraf}\fi}
\newtheorem{theorem}{Theorem}
\newtheorem{lemma}[theorem]{Lemma}

\newtheorem{proposition}[theorem]{Proposition}

\newenvironment{proof}{\begin{trivlist}\item[]{\flushleft\bf Proof }}
{\qed\end{trivlist}}

\begin{document}

\title{\bf Graph States, Pivot Minor,  and \\ Universality of $(X,Z)$-measurements}

\author{Mehdi Mhalla, Simon Perdrix}
%\inst{1,2} \email{Simon.Perdrix@imag.fr}

\date{\small CNRS,\\ LIG, University of Grenoble, France
}

\maketitle
\begin{abstract}  
The graph state formalism offers strong connections between quantum information processing and graph theory. 
Exploring these connections, first we show that any graph is a pivot-minor of a planar graph, and even a pivot minor of a triangular grid. Then, we prove that the application of measurements in the $(X,Z)$ plane (i.e. one-qubit measurement according to the basis $\{\cos(\theta) \ket 0+ \sin (\theta) \ket 1,\sin(\theta) \ket 0- \cos (\theta) \ket 1\}$ for some $\theta$) over graph states represented by triangular grids is a universal measurement-based model of quantum computation. These two results are in fact two sides of the same coin, the proof of  which is a combination of graph theoretical and quantum information techniques. 
\end{abstract}

\section{Introduction}

 In 2001, Raussendorf and Briegel \cite{RB01,RBB03} introduced a model for quantum computation based on measurements  where  one-qubit measurements 
 %according to $Z$ or in the $(X,Y)$-plane
  are applied on an initial entangled state, called  \emph{graph state}. %, are universal for quantum computation.
This model, called the \emph{one-way model} is not only a very promising model for the physical implementation of a quantum computer \cite{Wetal,Petal}, but it has also led to several theoretical breakthroughs in quantum information processing. For instance, the one-way model has been proved to be more favorable to the parallelization of quantum operations than the usual quantum circuits \cite{BKP10}; the one-way model has also given rise to the elaboration of several  protocols  like the blind quantum computing \cite{BFK08}, and the quantum secret sharing with graph states \cite{MS08,KMMP09,JMP11}. %This model is also been proved to be more favorable to the parallelization of quantum operations than the usual quantum circuits \cite{BKP}.

The graph state formalism, which is used to describe the initial entangled state in the one-way model, has been broadly studied this last decade, the survey by Hein  \emph{et~al.}
 \cite{HDERNB06-survey} provides an  excellent introduction to the domain and has more than 200 references. The graph state formalism is a powerful framework for characterizing quantum information properties in a combinatorial way, using graph theory. For instance, the ability of performing a unitary (or more generally an information preserving evolution) on a given graph state has been characterized as the existence of a certain kind of \emph{flow} in the corresponding graph \cite{DK06,DKPP09,BKMP07,MMPST11}.
 
Another example of the graphical characterization of quantum information properties is that two graphs  which are locally equivalent (i.e. equal up to local complementation, a graph transformation introduced by Kotzig \cite{kotzig} and investigated among others by Bouchet \cite{Fraysseix81,Bouchet85-connectivity}) are representing the same entanglement in the sense that the corresponding graph states are LC equivalent \cite{VdN04}. As a consequence, the rank-width \cite{OumSeymour} of a graph, which is invariant by local complementation,  is  a measure of entanglement of the corresponding quantum state \cite{VdNMDB06}. Moreover, the minimal degree up to local complementation \cite{HMP06} and the weak odd domination \cite{GJMP} are other examples of graph theoretical characterization of quantum properties: the minimal distance of a quantum code for the minimal degree up to local complementation \cite{ChuangShor}; and the threshold of a graph state based quantum secret sharing schemes for  weak odd domination \cite{GJMP}. 

In the present paper, we mainly prove two important results, both come from the strong connections, offered by the graph state formalism, between quantum information processing and graph theory. First we show that any graph is a pivot-minor of a planar graph (a pivot minor of a graph $G$ is a graph one can obtain by performing pivotings -- a certain combination of local complementations -- and vertex deletions over the graph $G$). Moreover, we prove that the application of measurements in the $(X,Z)$ plane (i.e. one-qubit measurement according to the basis $\{\cos(\theta) \ket 0+ \sin (\theta) \ket 1,\sin(\theta) \ket 0- \cos (\theta) \ket 1\}$ for some $\theta$) over graph states represented by triangular grids is a universal model of quantum computation. These two results are in fact two sides of the same coin, the proof of  which is a combination of graph theoretical and quantum information techniques. In particular, the former, i.e. the proof that any graph is a pivot-minor of a planar graph,  is an example of `classical' (in the sense `non quantum') property which is proved  using quantum arguments. Other such classical results with quantum proofs are listed in this paper \cite{DdW09}.

The latter result, the universality of $(X,Z)$-measurement over triangular grids, is an improvement in the quest of the minimal resources for measurement-based quantum computation \cite{N01,per05a,leu04,VdNMDB06,Per07a,BFK08,Taka10}, minimisation which is essential for the actual physical implementation of the model. Several regular grids (square, triangular, and hexagonal grids) are known to be universal resources for the one-way model \cite{VdNMDB06}, however the universality of these graph states is based on the use of single qubit measurements in the three possible axis $X$, $Y$ and $Z$ of the Bloch sphere, more precisely the performed measurements are according to $Z$ (measurement in the standard $\{\ket 0, \ket 1\}$ basis) and in the $(X,Y)$ plane (measurement in the basis $\{\frac{\ket 0+e^{i\theta}\ket 1}{\sqrt{2}},\frac{\ket 0-e^{i\theta}\ket 1}{\sqrt{2}}\}$ for some $\theta$). Alternatively, it has been proven in \cite{BFK08} that the measurements in the $(X,Y)$ plane over `brickwork' states (see Figure \ref{fig0}) are universal. The tradeoff to the use of the single $(X,Y)$ plane for the measurement is that the brickwork state is not regular, so not as easy as a grid to prepare. In the present paper, we show that measurements in a single plane, namely the $(X,Z)$ plane, over the triangular grid is universal.% The triangular grid is both regular and vertex transitive, and is, a consequence, a resource that is easier than the brickwork to prepare.  

\begin{figure}[h]
   \begin{center}
\includegraphics[width=8cm]{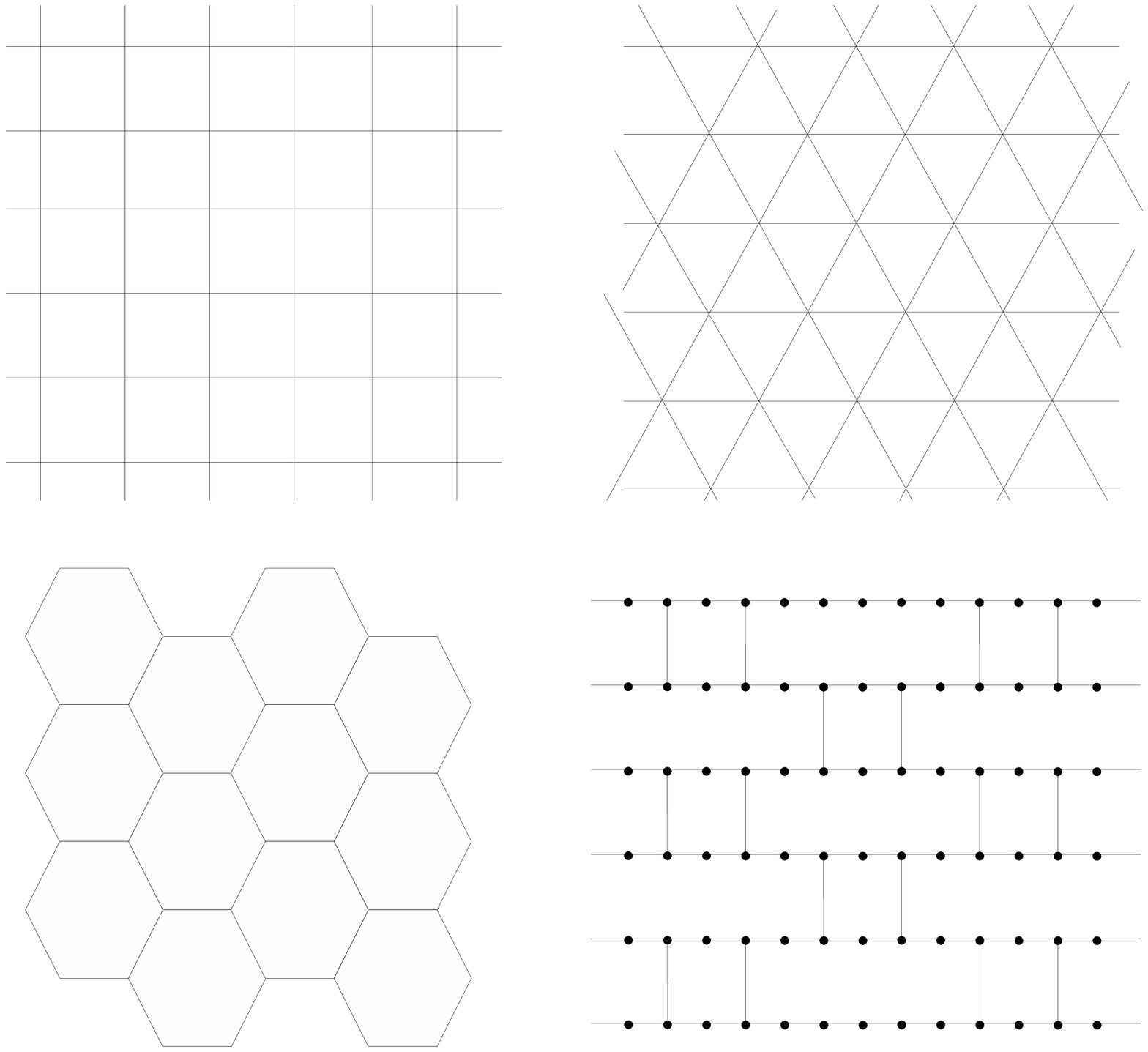}\vspace{-1cm}
\end{center}
\caption{\label{fig0} Top left: squarred grid (or cluster state) -- Top right: triangular grid  -- Bottom left: hexagonal grid  -- Bottom right: 'brickwork' state.}
\end{figure}

The paper is structured as follows: first, the  standard graph theoretical notions  used in the paper are presented, then are given the fundamental properties of the signed graph states, an extension of the graph state formalism.  %, and prove some combinatorial properties of graph states.
In section \ref{secmeas},  the actions of $X$- and $Z$-measurements over signed graph states are graphically characterized.
Section \ref{sec:pivotplanar} is dedicated to the proof, based on quantum arguments,  that any graph is pivot minor of a planar graph. %, and that the triangular grid plays a special role for $(X,Z)$ plane measurements comparing to the standard grid or the hexagonal grid. 
Finally, in section \ref{sec:univ}, we prove that measurements in the $(X,Z)$ plane on triangular grids are universal for quantum computing.
\section{Pivot minor}
\label{secpiv}

 Given a graph $G$ and a vertex $u$, the local complementation on the vertex $u$ switches edges and non-edges in the neighborhood of $u$: it transforms $G$ into $G*u=G \Delta K_{{\cal{N}}(u)}$ where $\Delta$ is the symmetric difference and  $ K_{{\cal{N}}(u)}$ is the complete graph over the neighborhood of $u$. 
Pivoting (also called edge local complementation) \cite{Bouchet87,Oum_pivot} on an edge $uv$ is defined as a sequence of local complementations on the two vertices of the edge $G \wedge uv=G*u*v*u$ (see Figure \ref{fg1}).

\psfrag{u}[c][][1]{$u$}
\psfrag{v}[c][][1]{$v$}
\psfrag{A}[c][][1]{$A\,$}

\psfrag{B}[c][][1]{$B$}
\psfrag{C}[c][][1]{$C$}
\psfrag{D}[c][][1]{$D$}

\begin{figure}[h]
   \begin{center}
\includegraphics[width=8cm]{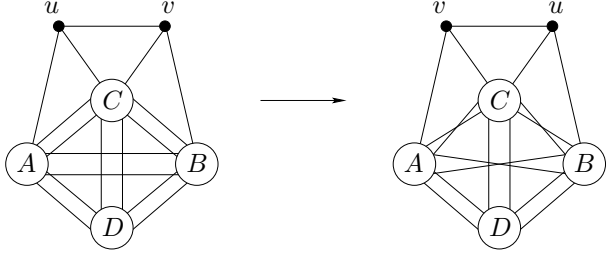}\vspace{-0.5cm}
\end{center}
\caption{\label{fg1}Pivoting on $uv$. $C={\cal N}(u)\cap {\cal N}(v)$, $A={\cal N}(u)\setminus C$, $B={\cal N}(v)\setminus C$,  and $D$ is the rest of the vertices. Pivoting on $uv$ exchanges vertices $u$ and $v$, and for any $(x,y) \in (A\times B)\cup(B\times C)\cup(A\times C)$, the edge $xy$ is deleted if $xy$ was an edge, and added otherwise.}
\end{figure}    
A graph $G$ is a pivot minor of $H$ if and only if $G$ can be obtained from $H$ by a sequence of pivotings and then a sequence of vertex deletions. 

It has been proven \cite{tuck60,Geelen97} %and Geelen97
 that  for any  sequence of pivotings, there exists an equivalent sequence of pivotings where each vertex  is used at most once. % The structural properties of pivoting is an active line of research \cite{BH11}.
  In particular, Kwon and Oum characterized, using the rank-width, the graphs which are pivot minors of  trees and paths \cite{ojoun11} and Oum
 proved that any bipartite circle graph is a pivot minor of all line graphs with large rank-width  \cite{oumpiv}. %The structure of pivot complementations is still being investigated \cite{BH11}.

\section{Signed graph state formalism}

In this section, we review the fundamentals of graph states \cite{HDERNB06-survey} and signed graph states \cite{DAB03}. The graph state formalism consists in representing some particular quantum states using graphs. Given a graph $G$ of order $n$, the corresponding quantum state is defined as the common fixpoint of $n$ operators depending on the  graph $G$. Each of these operators is a Pauli operator on $n$ qubits.  
The group of Pauli operators acting on a set $V$ of $n$ qubits is generated by 
$\{\op X_u, \op Z_u, i.I\}_{u\in V}$, where $I$ is the identity,  
$\op X_u$ (resp. $\op Z_u$) is an operator which acts as the identity on $V\setminus \{u\}$ and as $\op X : \ket x\mapsto \ket {\overline x}$ (resp. $\op Z : \ket x\mapsto (-1)^x \ket {x}$) on qubit $u$. 
More generally, for any subset $S\subseteq  V$, let $\op X_S =\prod_{u\in S} \op X_u$ and $\op Z_S = \prod_{u\in S} \op Z_u$.

\paragraph{}For a given simple undirected graph $G=(V,E)$ of order $n$, the \emph{graph state}  $\ket G\in \mathbb C^{n}$ is the unique quantum state\footnote{In fact $\ket G$ is unique up to a global phase which is irrelevant in quantum computing. This global phase is choosen s.t. $\bra {0^n} G\rangle = \frac 1{\sqrt{2^{n}}}$} such that for any $u\in V$, 
$$\op X_u \op Z_{\mathcal {\cal{N}}(u)} \ket G = \ket G$$% where $\mathcal \cal{N}_G(u)$ is the neighborhood of $u$ in $G$ (to simplify the notations $\cal{N}_G(u)$. 

To increase the expressive power of the graph state formalism, following \cite{DAB03} a \emph{sign} is added to the graph: for a given graph $G=(V,E)$ and a subset $S\subseteq V$ of vertices, let $$\ket {G;S} = \op Z_S \ket G$$

\begin{proposition} \label{prop:basis}
For any graph $G=(V,E)$, 
$\{\ket{G;S}\}_{S \subseteq V}$ 
form an orthonormal basis. 
\end{proposition} 

\begin{proof}
First,  notice that for any $u\in V$ and any $S\subseteq V$, $\op X_u$ and $\op  Z_{S}$ commute if $u\notin S$ and anticommute if $u\in S$. %o= (-1)^{|S\cap S'|}\op Z_{S'}\op X_S$. 
For any distinct $S,S'\subseteq V$, let $ u\in S \Delta S'$.  \begin{eqnarray*}\inner{G;S}{G;S'} &=& \bra G \op Z_{S} \op Z_{S'} \ket G\\ &=& \bra {G}\op Z_{S\Delta S'} \op X_u \op Z_{\mathcal {\cal{N}}(u)} \ket G \\&=&
-  \bra {G} \op X_u \op Z_{\mathcal {\cal{N}}(u)} \op Z_{S\Delta S'}\ket G \\ &=& - \inner{G;S}{G;S'} \end{eqnarray*}
As a consequence, $\inner{G;S}{G;S'} =0$, so $\{\ket{G;S}\}_{S \subseteq V}$ 
is an orthonormal basis. 
\end{proof}

In the following lemma, it is  shown  that the action of any Pauli operator on a signed graph state can be captured by its sign, up to a global phase. Since quantum state are equivalent up to a global phase, for any $\ket \phi, \ket \psi$, we write $\ket \phi \equiv \ket \psi$ when there exists $\alpha$ such that $\ket \phi = e^{i\alpha}\ket \psi$.

\begin{lemma} 
For any graph $G=(V,E)$, any subset $S \subseteq V$, and any Pauli operator $\op P$, there exists $S'\subseteq V$ such that $\op P\ket {G; S} \equiv \ket {G; S'}$.% (up to a global phase). 
\end{lemma} 
\begin{proof}
$\op P$ can be decomposed, up to a global phase, into a product of $\op X$ and $\op Z$ operators: $\exists S_1,S_{2} \subseteq V,$ and $d\in [0,3]$ such that  $\op P=i^d\op X_{S_{1}}\op Z_{S_{2}}$. Thus $\op P\ket{G;S} = i^d \op X_{S_{1}} \op Z_{S_2\Delta S}\ket G = i^d X_{S_{1}}\op Z_{S_2\Delta S} (\Pi_{u\in S_{1}}\op X_{u}\op Z_{{\cal{N}}(u)}\ket G) = i^{d'}\op Z_{S_{2}\Delta S}\Pi_{u\in S_{1}} \op Z_{{\cal{N}}(u)}\ket G$, where $\Delta$ is the symmetric difference.
\end{proof}

All signed graph states represent distinct quantum states:

\begin{lemma}\label{lem:distinct}
For any graphs $G_{1}, G_{2}$ and any signs $S_{1}, S_{2}$, \begin{equation*} \ket{G_{1};S_{1}}\equiv\ket{G_{2};S_{2}} \implies G_{1}=G_{2} \text{ ~and~ } S_{1}=S_{2} \end{equation*}
\end{lemma}

\begin{proof} Let $S=S_{1}\Delta S_{2}$. $ \ket{G_{1};S_{1}}\equiv \ket{G_{2};S_{2}}$ implies $ \ket{G_{1};S}\equiv \ket{G_{2}}$. For any vertex $u$, since $\op Z_{S} \op X_{u}\op Z_{{\cal{N}}_{G_{1}}(u)} \ket {G_{1}}\equiv \op X_{u}\op Z_{{\cal{N}}_{G_{2}}(u)} \ket {G_{2}}$, $\ket{G_{1}; S} \equiv  \ket {G_{2};{\cal{N}}_{G_{1}(u)}\Delta {\cal{N}}_{G_{2}(u)}}$. Thus, $\ket {G_{2}} \equiv \ket {G_{2};{\cal{N}}_{G_{1}(u)}\Delta {\cal{N}}_{G_{2}(u)}}$. By proposition \ref{prop:basis}, ${\cal{N}}_{G_{1}(u)}\Delta {\cal{N}}_{G_{2}(u)}=\emptyset$. Thus, $G_{1}=G_{2}$ since for any $u$, ${\cal{N}}_{G_{1}}(u) =  {\cal{N}}_{G_{2}}(u)$. As a consequence,  $ \ket{G_{1};S}\equiv \ket{G_{1}}$, so by proposition \ref{prop:basis}, $S=\emptyset$ so $S_{1}=S_{2}$.
\end{proof}

\section{Combinatorial properties of Graph States} 

The success of the graph state formalism is mainly due to the ability to characterize quantum properties by means of graph theoretic ones. For instance, local Clifford equivalence has been characterized by local complementation \cite{VdN04}. Here, we recall that the action of local Clifford transformations
 can be characterized by local complementations, and we prove that the action of local real Clifford transformations can be characterized by pivoting. %Moreover, extension of these characterizations to local unitary, and local real unitary respectively, are conjectured. 
 
 %\subsection{Local unitary transformations}

 A Clifford transformation $\op C$ is a map which transforms Pauli  operators to Pauli operators: for any Pauli $\op P$, $\op {CPC}^\dagger$ is a Pauli operator. A local Clifford is a Clifford that can be decomposed into the tensor product of one-qubit unitaries: $\op C_1\otimes \ldots \otimes \op C_n$. The group of local Clifford on $V$ is generated by $\{ \sqrt \op Z_u, \sqrt \op X_u\}_{u\in V}$ where $\sqrt {\op P}_u := \frac {e^{i\pi/4}}{\sqrt{2}}(I-i\op P_u)$. For any subset $S$ of qubits, let $\sqrt {\op P}_S := \prod_{u\in S} \sqrt {\op P}_u$. 

A particularly important Clifford operation is the so-called Hadamard transformation: for any $u\in V$ 
$\op H$ acts as the identity on $V\setminus \{u\}$ and as $\op H : \ket x \mapsto\frac{\ket 0+(-1)^x\ket 1}{\sqrt 2}$ on $u$. 
  %and any $x\in \{0,1\}^V$: 
Moreover, for any $S\subseteq V$,  let $\op H_S = \prod_{u\in S} \op H_u$.

$\op H$ together with $\op Z$ generates a subgroup of the Clifford operations, which corresponds exactly to the local \emph{real} Clifford operations, i.e. those that can be represented by matrices with real entries in the computational basis $\{\ket 0, \ket 1\}$.

Some quantum transformations on graph states can be interpreted in terms of graph transformations. In particular, the application of some Clifford transformations can be interpreted in terms of local complementions  \cite{VdN04}: 

\begin{proposition}
For any graph $G=(V,E)$, for any vertex $u\in V$, $$\sqrt \op X^\dagger_{u}\sqrt \op Z_{{\cal{N}}(u)} \ket G = \ket{G*u}$$
\end{proposition}

Such Clifford operators act on signed graph states as follows:% signVdN:edge-local-equivalenceed graph states: 

$$\sqrt \op X^{\dagger}_{u} \sqrt \op Z_{{\cal{N}}(u)} \ket {G;S}=\begin{cases}\ket {G*u;S}& \text{ if $u\notin S$} \\ \ket {G*u;S\Delta {\cal{N}}(u)}& \text{ if $u\in S$} \\ \end{cases}$$

Moreover Van den Nest proved that the action of local Clifford transformations is characterized by local complementation:
\begin{lemma}[\cite{VdN04}]
For any graphs $G$ and $G'$, there exists a local Clifford transformation $\op C$ such that $\op C \ket G = \ket {G'}$ if and only if $G$ and $G'$ are locally equivalent, i.e. there exists a sequence of local complementation transforming $G$ in $G'$. 
\end{lemma}

In \cite{VdNedge-local-equivalence},  a combinatorial characterizationg of the action of Hadamard transformations on two neighbor qubits has been introduced\footnote{In  \cite{VdNedge-local-equivalence} the characterisation has been given up to Pauli operation.}.
\begin{proposition} 
For any graph $G=(V,E)$ and any edge $(u,v)\in E$, 
$$ \op H_{u,v}\ket G = \op Z_{{\cal{N}}(u)\cap {{\cal{N}}(v)}}\ket {G\wedge uv}$$
\end{proposition}

\begin{proof} Since the original statement in \cite{VdNedge-local-equivalence} has been introduced without the Pauli factor $\op Z_{{\cal{N}}(u)\cap {\cal{N}}(v)}$, a proof of the proposition is given.  
The proof is based on the facts that $\sqrt\op X_u^\dagger \sqrt \op Z_u \sqrt \op X_u^\dagger = e^{i\pi/4}\op X_u\op H_u \op X_u$ and $\sqrt\op Z_u \sqrt \op X_u^\dagger  \sqrt \op X_u = e^{-i\pi/4}\op Z_u\op H_u \op Z_u$.
$$\begin{array}{rcl}
\ket{G\wedge uv}&=&\ket {G*u*v*u}\\
&=&\op Z_{(\mathcal \cal{N}(u)\cup \mathcal N(v))\Delta \{u,v\}}\sqrt\op X_u^\dagger \sqrt \op Z_u \sqrt \op X_u^\dagger          \sqrt\op Z_v \sqrt \op X_v^\dagger  \sqrt \op X_v  \ket G\\
&=&\op Z_{(\mathcal N(u)\cup \mathcal N(v))\Delta  \{u,v\}} \op X_u \op Z_v \op H_u \op H_v \op X_u \op Z_v \ket G\\
&=&\op X_u \op \op Z_{(\mathcal N(u)\cup \mathcal N (v))\Delta \{u\}} \op H_{u,v} \op Z_{\mathcal N(u)\Delta \{v\} }  \ket G\\
\end{array}$$%}
Moreover $\ket{G\wedge uv} = \op X_u\op Z_{N_{G\wedge uv}(u)}\ket{G\wedge uv}= \op X_u \op Z_{N_{G}(v)\Delta \{u,v\}}\ket{G\wedge uv}$, so $\ket{G\wedge uv} = \op H_{u,v} \op Z_{\mathcal N_G(u)\cap \mathcal N_G(v)}\ket G$.
\end{proof}

 The action of Hadamard transformations is extented to signed graph states: Given a graph $G=(V,E)$ and an edge $(u,v)\in E$, for any $S\subseteq V\backslash \{u,v\}$: 
$$\begin{array}{lcr}
\op H_{u,v} \ket {G;S}&=&\ket {G\wedge uv;S\Delta (\mathcal N(u)\cap \mathcal N(v))}\\
\op H_{u,v} \ket {G;S\cup \{u\}}&=&\ket {G\wedge uv;S\Delta (\mathcal N(u)\cap \overline{\mathcal N(v)})}\\
\op H_{u,v} \ket {G;S\cup \{v\} }&=&\ket {G\wedge uv;S\Delta (\overline{\mathcal N(u)}\cap \mathcal N(v))}\\
\op H_{u,v} \ket {G;S\cup \{u,v\}}&=&-\ket {G\wedge uv;S\Delta (\mathcal N(u)\cup \mathcal N(v))}
\end{array}$$

Now we prove that  the action of  the local real Clifford transformations on graph states is captured by the graphical transformation of pivoting. 

\begin{lemma}\label{lem:realcliff}
For any real local Clifford transformation $R$, and for any graphs $G_1=(V_1,E_1)$ and $G_2=(V_2,E_2)$, if $R\ket {G_1,S_1}\equiv \ket {G_2,S_2}$ then $\exists A,S_1', S_2'\subseteq V_1$ s.t. $\op H_A\ket {G_1,S_1'}\equiv \ket {G_2,S_2'}$.
\end{lemma}

\begin{proof}
Any real local Clifford can be decomposed (up to a global phase) into a product of $\op H$, $\op X$ and $\op Z$. More precisely, for any real local Clifford $R$, $\exists A,B,C$ s.t. $R\equiv \op H_A\op X_B\op Z_C$. So $R\ket {G_1,S_1}\equiv \op H_A\op X_B\op Z_C  \ket {G_1,S_1}\equiv  \op H_A \op X_B \ket {G_1,S_1\Delta C} \equiv \op H_A \ket {G_1,S_1\Delta C \Delta Odd(B)}  \equiv \ket {G_2,S_2}$ where $Odd(B):=\{v ~s.t. |N(v)\cap B|=1Ê\mod 2\}$\footnote{$Odd(B)$ can be alternatively defined inductively as follows: $Odd_G(\emptyset) := \emptyset$, and $Odd_G(\{u_0,\ldots u_n\}) = Odd_G(\{u_0,\ldots, u_{n-1}\})\Delta N_G(u_n)$}.
\end{proof}

\begin{lemma}\label{lem:XZpivot}
For any graphs $G$ and $G'$,  any signs $S$, $S'$, and any real local Clifford transformation $R$, if $R \ket {G;S} = \ket {G';S'}$ then $G$ and $G'$ are pivot equivalent, i.e. there exists a sequence of pivotings transforming $G$ in $G'$.
\end{lemma}

\begin{proof}
Thanks to Lemma \ref{lem:realcliff}, there exist $A,S_1,S_2$ s.t. $\op H_A\ket {G,S_1} \equiv \ket {G',S_2}$.
 For any $u\in A$, notice that $A\cap \mathcal N(u)\neq \emptyset$. Otherwise, if $A\cap N_G(u)=\emptyset$ then $\ket {G';S_2}\equiv \op H_{A}\op Z_{S_1} \op X_{u}\op Z_{\mathcal N(u)}\ket G\equiv   \op Z_{u}\op Z_{\mathcal N(u)} H_{A}\op Z_{S_1} \ket G\equiv   \op Z_{\{u\}\cup \mathcal N(u)}\ket {G';S_2}$ which contradicts proposition \ref{prop:basis}. Thus, for any $u\in A$, there exists $v\in A\cap \mathcal N(u)$ and a sign $S_3$ such that $\op H_{A}\ket {G;S_1}\equiv \op H_{A\backslash\{u,v\}} \op H_{u,v}\ket {G;S_1}\equiv \op H_{A\backslash \{u,v\}}\ket {G\wedge uv;S_3}\equiv \ket{G';S_2}$. As a consequence, by induction on the size of $A$, $G$ is a pivot equivalent to $G'$.
\end{proof}

\section{Local measurements and Measurement-based quantum computing}
\label{secmeas}

In this section, the action of local measurements over the qubits of a graph states is considered. We give graphical interpretation of such measurements, but also a computational interpretation since the application of local measurement over a graph state is the key ingredient of the one-way model. 

We consider local measurements according to the basis $\{\ket{0^{(\alpha)}}, \ket{1^{(\alpha)}}\}$ parametrised by an angle $\alpha \in [0, 2\pi)$, where  
\begin{eqnarray*}
\ket{0^{(\alpha)}} &:= &\cos(\frac \alpha 2) \ket 0+ \sin (\frac \alpha 2) \ket 1\\
\ket{1^{(\alpha)}}&:= &\sin(\frac \alpha 2) \ket 0- \cos (\frac \alpha 2)\ket 1
 \end{eqnarray*}

This family of measurements is said to be in the $(X,Z)$-plane since all these basis states are actually in the $(X,Z)$-plane in the Bloch sphere representation. 
Moreover,  for any $\alpha$, the observable associated with the measurement in the basis $\{\ket{0^{(\alpha)}} , \ket{1^{(\alpha)}} \}$ is $\cos(\alpha)\op Z+\sin(\alpha)\op X$. 
Notice that when $\alpha=0$ it corresponds to the standard basis measurement $\{\ket 0, \ket 1\}$  also called  $\op Z$-measurement and when $\alpha = \pi/2$ to the diagonal basis measurement $\{ \frac{\ket 0+\ket 1}{\sqrt 2}, \frac{\ket 0-\ket 1}{\sqrt 2}\}$ or $\op X$-measurement.

One-way quantum computation is generally made upon $(X,Y)$-measurements together with $Z$-measurements. We prove that using $(X,Z)$ measurements instead also give rise to a universal model of quantum computing.

A measurement has non deterministic evolution which consists in projecting the state of the measured qubit onto one of the two basis states: $\ket {s^{(\alpha)}}$, $s\in \{0,1\}$ is the classical outcome of the measurement. Thus the action of such a measurement, up a re-normalisation is $\bra {s^{(\alpha)}}$.

The action of a $\op Z$-measurement corresponds to a vertex deletion:
for any graph $G=(V,E)$,  any vertex $u\in V$, any $s  \in \{0,1\}$,

$$\sqrt 2 ~\bra {s_u^{(0)}}  G \rangle= \op Z^s_{ \mathcal N(u)}\ket {G\backslash u}$$

The action of $\op Z$-measurement is extended to signed graph states: for any graph $G=(V,E)$, any sign $S\subseteq V$, any vertex $u\in V$, and any $s\in \{0,1\}$, 

$$\sqrt 2 ~\bra {s_u^{(0)}}  {G;S} \rangle= \ket {G\backslash u; (S\backslash u)\Delta (\mathcal N(u))^s}$$

A single $X$-measurement cannot be interpreted as a graph transformation, since for any connected graph $G$ and any vertex $u$ of $G$, an $X$-measurement of qubit $u$ of $\ket G$ leads to a quantum state which is no more a graph state. However, the application of a pair of $X$-measurements on adjacent qubits can be interpreted as a pivoting followed by vertex deletions:

\begin{proposition}\label{prop:Xpiv}
For any graph $G$, any edge $uv$, and any $r,s \in \{0,1\}$, 

$$2~\bra{s_u^{(\pi/2)}r_v^{(\pi/2)}} G\rangle = \op Z_{S}\ket {G\wedge uv \backslash u\backslash v}$$
where $S= (\mathcal N(u)\cap \mathcal N(v))\Delta (\mathcal N(u)\backslash v)^r \Delta (\mathcal N(v)\backslash u)^s$

\end{proposition}

\begin{figure}[h]
   \begin{center}
\includegraphics[width=8cm]{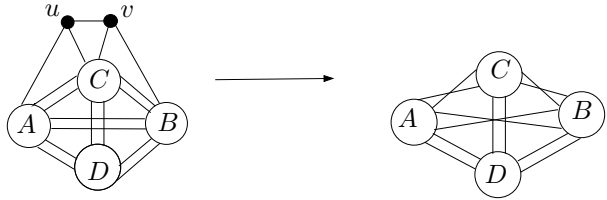}
\end{center}
\caption{\label{fg2}Evolution of a graph state by $X$ measuring $u$ and $v$}
\end{figure}  

\begin{proof}
Since $\bra {s_u^{(\frac\pi2)} }= \bra {s_u^{(0)}}\op H_u$, $
2\bra{s_u^{(\frac\pi2)}r_v^{(\frac\pi2)}} G\rangle = 2\bra{s_u^{(0)}r_v^{(0)}} \op H_{u,v} \ket G$\\$= 2~\bra{s_u^{(0)}r_v^{(0)}} \op Z_{N_G(u)\cap N_G(v)}\ket {G\wedge uv}$\\ $= \op Z_{(N_G(u)\cap N_G(v))\Delta (N_{G\wedge uv}(v)\backslash u)^r\Delta  (N_{G\wedge uv}(v)\backslash v)^r}\ket {G\wedge uv\backslash u\backslash v}$\\ $=\op Z_{(N_G(u)\cap N_G(v))\Delta (N_{G}(u)\backslash v)^r\Delta  (N_{G}(v)\backslash u)^r}\ket {G\wedge uv\backslash u\backslash v}$
\end{proof}

The action of the $X$-measurement of two neighbors extends to the signed graphs: for any graph $G=(V,E)$, any sign $S\subseteq V$, any edge $uv\in E$ and any $r,s\in \{0,1\}$:

$$2~\bra{s_u^{(\pi/2)}r_v^{(\pi/2)}} {G;S}\rangle = \ket {G\wedge uv \backslash u\backslash v; S'}$$
where $S'= (S\backslash u \backslash v)\Delta (\mathcal N(u)\cap \mathcal N(v))\Delta (\mathcal N(u)\backslash v)^{r\oplus |v\cap S|} \Delta (\mathcal N(v)\backslash u)^{s\oplus |u\cap S|}$

\begin{lemma}
For any graphs $G$, $G'$, and any signs $S$, $S'$, if a sequence of $X$- and $Z$-measurements  transforms $\ket {G;S}$ into $\ket {G';S'}$ then $G'$ is a pivot minor of $G$.
\end{lemma}

\begin{proof}
Since all measurements are commuting because they are local, we assume w.l.o.g.  that all the $Z$-measurements are performed first, leading to a signed graph state $\ket {G\backslash A,  S_1}$, where $A$ is the set of the qubits which are $Z$-measured. Then, inductively,  if a pair of 2 neighbors $uv$ are $X$-measured then by proposition \ref{prop:Xpiv}, these two measurements lead to the sign graph state $\ket {((G\backslash A)\wedge uv)\backslash \{u,v\},  S_2}$. We repeat this step inductively until there is no two neighbors which are $X$-measured. So it leads to a graph state $\ket {\tilde G,S_3}$ such that  $\tilde G$ is a pivot minor of $G$ and such  that  the remaining $X$-measurements have to be performed on an independent set $B$ of $\tilde G$. If $B$ is empty then $G'=\tilde G$ so $G'$ is a pivot minor of $G$. Otherwise, let $u\in B$. The state $\ket \phi$ after the $X$-measurement of $u$ is $\sqrt 2. \bra {s^{(\pi/2)}_u}{\tilde G; S_3}\rangle$ with $s\in \{0,1\}$. $\ket \phi$ is an eigenvector of $\op Z_{\mathcal N_{\tilde G}(u)}$, indeed $\op Z_{\mathcal N_{\tilde G}(u)}(u)  \bra {s^{(\pi/2)}_u} {\tilde G; S_3}\rangle = (-1)^s \op Z_{\mathcal N_{\tilde G}(u)}(u) \bra {s^{(\pi/2)}_u} \op X_u \op Z_{\mathcal N_{\tilde G}(u)}|{\tilde G; S_3}\rangle = (-1)^s  \bra {s^{(\pi/2)}_u} {\tilde G; S_3}\rangle$. Moreover, since none of the neighbors of $u$ are measured, the final state $\ket {G',S'}$ is also a eigenvector of $\op Z_{\mathcal N_{\tilde G}(u)}$, so according to lemma \ref{lem:distinct} $\mathcal N_{\tilde G}(u)=\emptyset$ and $u$ is isolated. For an isolated vertex $u$, one can easily show that a $X$-measurement of $u$  produces the state $\ket{\tilde G\backslash u, S_3}$. By induction on the independent set $B$, it comes that $G'$ is a pivot minor of $G$.%Since the mesurements are action on distinct qubits, one can assume 
\end{proof}

\section{From Quantum computation to graph theory} \label{sec:pivotplanar}

\subsection{Embedding in a planar graph.}

The main result of this section is that any graph is pivot minor of a planar graph.
The proof consists in encoding the preparation of $\ket G$, for any $G$,  with a measurement-calculus pattern \cite{DKP07,DKPP09} (the language used for formal description of one-way quantum computation) composed of $X$-measurements only. 

First we recall some properties of the measurement-based  model:

 An open graph is a triplet $(G,I,O)$ where $G=(V,E)$ is a graph and $I,O\subseteq V$ are representing respectively the input and the output qubits of the computation. Intuitively, the input qubits are initialised in a state which is the state of the input of the computation and then they are entangled with the rest of the qubits. Then, the non output qubits are measured and finally, at the end of the computation, the output of the computation is located on the remaining qubits, i.e. the output qubits. For a given open graph $(G,I,O)$ the corresponding entangling operation $\op M_{(G,I,O)}$ is:

$$\op M_{(G,I,O)} :  \ket x \mapsto  \frac 1{\sqrt{2^{|I|}}}\sum_{S\subseteq I}(-1)^{x\bullet \den S}\ket {G;S}$$
 where $\den S$ is a binary $n$-bit vector indexed by $V$ such that $\den S(u)=1 \iff u\in S$, and $\forall x, y \in \{0,1\}^V$, $x\bullet y = \sum_{u\in V}x_uy_u$. \\
When the input state is the uniform superposition $\ket \phi = \frac 1{\sqrt{2^{|I|}}}\sum_{x\in \{0,1\}^I} \ket x$, the graph state $\ket G$ is prepared, indeed
$\bra G \op M_{(G,I,O)}\ket \phi =\frac 1{{2^{|I|}}}\sum_{S\subseteq I, x\in \{0,1\}^I}(-1)^{x\bullet \den S} \inner G {G;S} =\frac 1{{2^{|I|}}}\sum_{x\in \{0,1\}^I}(-1)^{x\bullet \den \emptyset} = 1$. 

\paragraph{} Open graphs can be composed: the composition $(G_2=(V_2,E_2),A,O)\circ (G_1=(V_1,E_1),I,A) = (G'=(V_1\cup V_2,E_1 \Delta E_2),I,O)$, %where $V(G'):= V(G_1) \cup V(G_2)$ and $E(G'):= E(G_1)\Delta E(G_2)$ 
is well defined if $V_1\cap V_2 = A$. 

In the following, we consider that the input qubits are always $\op X$-measured. The other qubits are measured in the $(X,Z)$-plane.

In order to simulate a unitary tranformation using measurements a correction strategy needs to be given to guarantee that the result of the simulation does not depend on the various classical outcomes abtained during the simulation. This strategy consists in applying the unitary maps Pauli $\op X$ or $\op Z$  on unmeasured qubits when the outcome of a measurement is $1$ (see \cite{MP08-icalp} for a detailed correction strategy). These corrections are applied in such a way that the state of the system after the correction is the same state  as the state of the system if the classical outcome $0$ would have occured. 
One of the simplest example of MBQC consist in the graph $P_2$ %(Figure \ref{HP2})
 composed of two vertices $1$ and $2$ connected by an edge. Notice that the pivoting according to the edge $(1,2)$ does not change the state: $P_2 \wedge 12 = P_2$, so $\op H_{1,2} \ket {P_2} = \ket {P_2}$. 

\begin{lemma}\label{lem:H}
The open graph $(P_n, \{1\},\{n\})$ where the qubits $1\ldots n-1 $ are  $X$-measured simulates the Hadamard transformation if $n$ is even and the identity if $n$ is odd.
 \end{lemma}
\begin{proof}
\begin{eqnarray*}
\sqrt 2~ \bra {r^{(\pi/2)}}\op M_{(P_2,\{1\},\{2\})}\ket x&=&\sqrt 2~ \bra {r^{(\pi/2)}} \frac 1{\sqrt{2^{|\{1\}|}}}\sum_{S\subseteq \{1\}}(-1)^{x\bullet \den S}\ket {P_2;S}\\
&=&\bra {r_1^{(\pi/2)}} {P_2}\rangle+(-1)^x\bra {r_1^{(\pi/2)}} \op Z_1\ket {P_2}\\
&=&\bra {r_1^{(\pi/2)}} {P_2}\rangle+(-1)^x\bra {r_1^{(\pi/2)}} \op X_2\ket {P_2}\\
&=&(\op I+(-1)^x\op X_2)\bra {r_1^{(\pi/2)}} {P_2}\rangle\\
&=&(\op I+(-1)^x\op X_2)\bra {r_1^{(\pi/2)}} \op H_{1,2}\ket {P_2}\\%+(-1)^x\bra {r_1^{(\pi/2+\pi)}}  \op H_{1,2}\ket {P_2}\\
&=&(\op I+(-1)^x\op X_2)\op H_2\bra {r_1^{(0)}} {P_2}\rangle\\
&=&\op H_2(\op I+(-1)^x\op Z_2)\op Z^r_2 \frac{\ket {0_2}+\ket {1_2}}{ 2} \\
&=&\op X^r_2\op H_{2}\ket x\\
\end{eqnarray*}
So it simulates $\op H$ up to a Pauli operator $\op X$ which depends on the classical outcome of the measurement.

By composing this elementary open graph state, it comes that any path of size $n$ where all but the output qubit is X-measured, implements $H$ if the number of edges is odd and the identity otherwise, since $H^2=I$. 
\end{proof}

\begin{lemma}\label{lem:CZ}
The open graph $(P_2,\{1,2\},\{1,2\})$ with no measurement, simulates the two-qubit unitary transformation $\Lambda Z$:  $\ket {x,y} \mapsto (-1)^{xy} \ket {x,y}$
\end{lemma}
\begin{proof}
The simulation of $\Lambda Z$ is obtained by considering a the open graph $(P_2,\{1,2\},\{1,2\})$. Thus both qubits are inputs and outputs. The simulation is done without measurement. For any $x,y\in \{0,1\}^{\{1,2\}}$, 
\begin{eqnarray*}
\bra y \op M_{(P_2,\{1,2\},\{1,2\})}\ket x &=& \frac 12\sum_{S\subseteq \{1,2\}}(-1)^{x\bullet 1_S}\bra y  {P_2;S}\rangle \\
&=& \frac 12 \sum_{S\subseteq \{1,2\}}(-1)^{x\bullet 1_S}\bra y \op Z_S \ket {P_2}\\
&=& \frac 12 \sum_{S\subseteq \{1,2\}}(-1)^{x\bullet 1_S}(-1)^{y\bullet 1_S}\bra y {P_2}\rangle\\
&=& \frac 12 \left(\sum_{S\subseteq \{1,2\}}(-1)^{(x\oplus y)\bullet 1_S}\right)\bra y {P_2}\rangle
\end{eqnarray*}

If $x\neq y$, then $\sum_{S\subseteq \{1,2\}}(-1)^{(x\oplus y)\bullet 1_S} = 0$, so $\bra y \op M_{(P_2,\{1,2\},\{1,2\})}\ket x=0$. When $x=y$, $\sum_{S\subseteq \{1,2\}}(-1)^{(x\oplus y)\bullet 1_S} = 4$. As a consequence, 
For any $x\in \{0,1\}^{\{1,2\}}$, 
\begin{eqnarray*}
\bra x \op M_{(P_2,\{1,2\},\{1,2\})}\ket x &=& 2~ \bra x {P_2}\rangle\\
&=&2~\bra{ x_1^{(0)}}\bra {x_2^{0}}{P_2}\rangle\\
&=& \bra{ x_1^{(0)}}({\ket {0_1}+(-1)^{x_2}\ket {1_1}})\\
&=&(-1)^{x_1.x_2}
\end{eqnarray*}

So $\op M_{(P_2,\{1,2\},\{1,1\})}$ implements the map $\ket {x_1x_2}\mapsto (-1)^{x_1.x_2}\ket {x_1x_2}$, i.e. $\Lambda Z$. 
\end{proof}

An other property required for the proof of the main theorem is the existence of a planar circuit preparing the state $\ket{G}$.
\begin{lemma}
\label{lem:univ}
For any  graph $G=(V,E)$ of order $n=|V|$, there exists a $n$-qubit \emph{planar} circuit   $C_G$ of size $O(n^3)$ composed of  $\op {\Lambda Z}$ and $\op H$ to prepare $\ket G$.%, where $\op{\Lambda Z}$ is the two-qubit unitary map: $\forall u,v\in V$, $\op{\Lambda Z}_{u,v} : \ket x \mapsto (-1)^{x_ux_v} \ket x$. 
\end{lemma}
\begin{proof} 
There exists a circuit of size $O(n^2)$ composed of $\op {\Lambda Z}$ only for preparing the graph state $\ket G$ \cite{HDERNB06-survey}. This circuit can be made planar by interspersing at most $n$ $\op{SWAP}$ gates between every $\op{\Lambda Z}$. Moreover every  $\op {SWAP}$ gate can be decomposed using  a constant number of $\op {\Lambda Z}$ et $\op H$:  $\forall u,v$, 
$$\op {SWAP}_{u,v} = \op H_{u} \op {\Lambda Z}_{u,v} \op H_{u,v} \op {\Lambda Z}_{u,v} \op H_{u,v} \op {\Lambda Z}_{u,v} \op H_{u}$$
As a consequence $\ket G$ can be prepared by a planar circuit of size $O(n^3)$ acting on $n$ qubits. 
\end{proof}
\begin{theorem}
 Any graph on $n$ vertices is pivot minor of a planar graph of $O(n^3)$ vertices. 
\end{theorem}
\begin{proof}
For any graph $G$, the planar circuit which implements $G$ on imput $\frac{1}{\sqrt{2^{|V|}}}\sum_{x\in \{0,1\}^{|V|}} \ket x$ can be simulated using an open graph $G'$ on with all but output qubits are $X$-measured: this open graph is obtained by composing the open graphs simulating $H$ and $\Lambda Z$ accordingly to the circuit. Since only $X$ measurements occur, $G$ is a pivot minor of $H$ (lemma \ref{lem:XZpivot}). Therefore any graph is a pivot minor of a planar graph.
\end{proof}

\subsection{Some particular planar graphs}

\begin{theorem}\label{thm:pivotgrid}
Any graph of $n$ vertices is pivot minor of a triangular grid on $O(n^4)$ vertices.
\end{theorem}
\begin{proof}
 %we show how
  The simulation of the circuit $C_G$ that prepares $\ket{G}$ can be embedded into a triangular grid of size $4n*4d$ where $n= |V|$ and $d$ is the depth $C_G$ :

The simulation with $X$ and $Z$  measurement of any circuit using the $\Lambda Z$ and $H$ gates 
is explained by the following figures where  the non-output vertices on the bold line are $X$-measured  and the other non-output vertices  are  $Z$-measured.  Note that it is important to be able to simulate the identity in this model in order to compose the different simulations.

The output qubits   (represented with a white circle) state corresponds to the  application of the simulated  unitary  transformation on the input qubits. (represented by a square).

\begin{itemize}
	\item Simulation of $Id$:
		\begin{center}
			\includegraphics[scale=0.15]{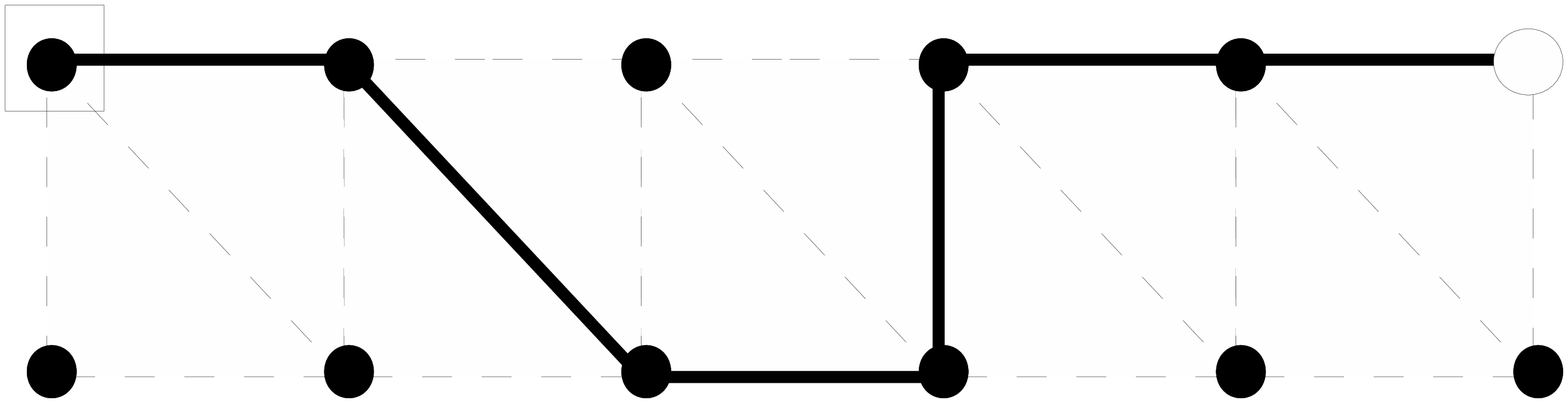} 
			\vspace{-0.8cm}
		\end{center}
	\item Simulation of $H$:
		\begin{center}
			\includegraphics[scale=0.15]{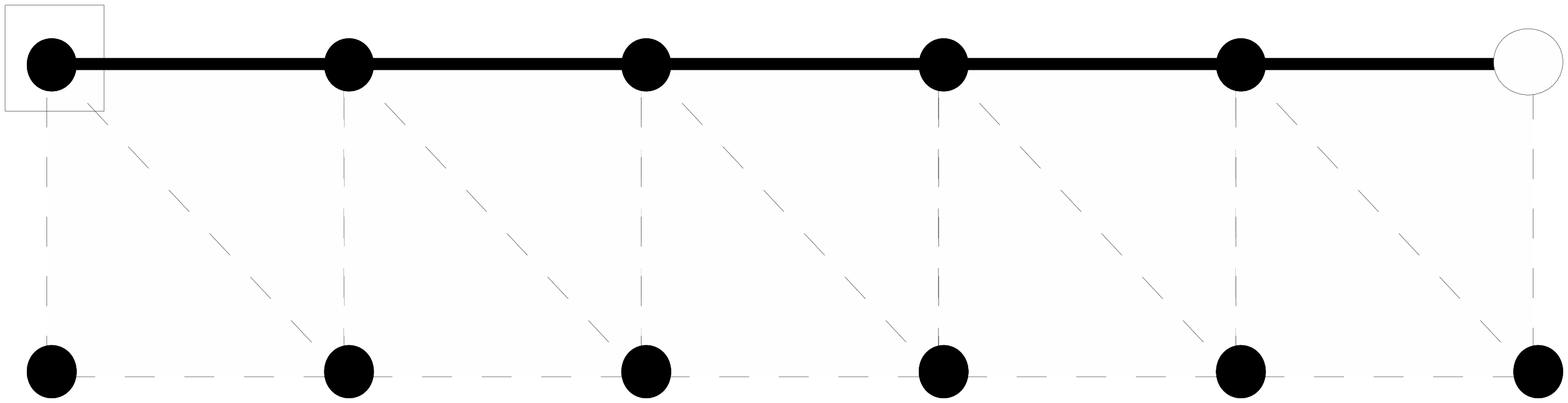} 
			\vspace{-0.8cm}
		\end{center}	
	\item Simulation of $\Lambda Z$:
		\begin{center}
			~~~~~~~\includegraphics[scale=0.23]{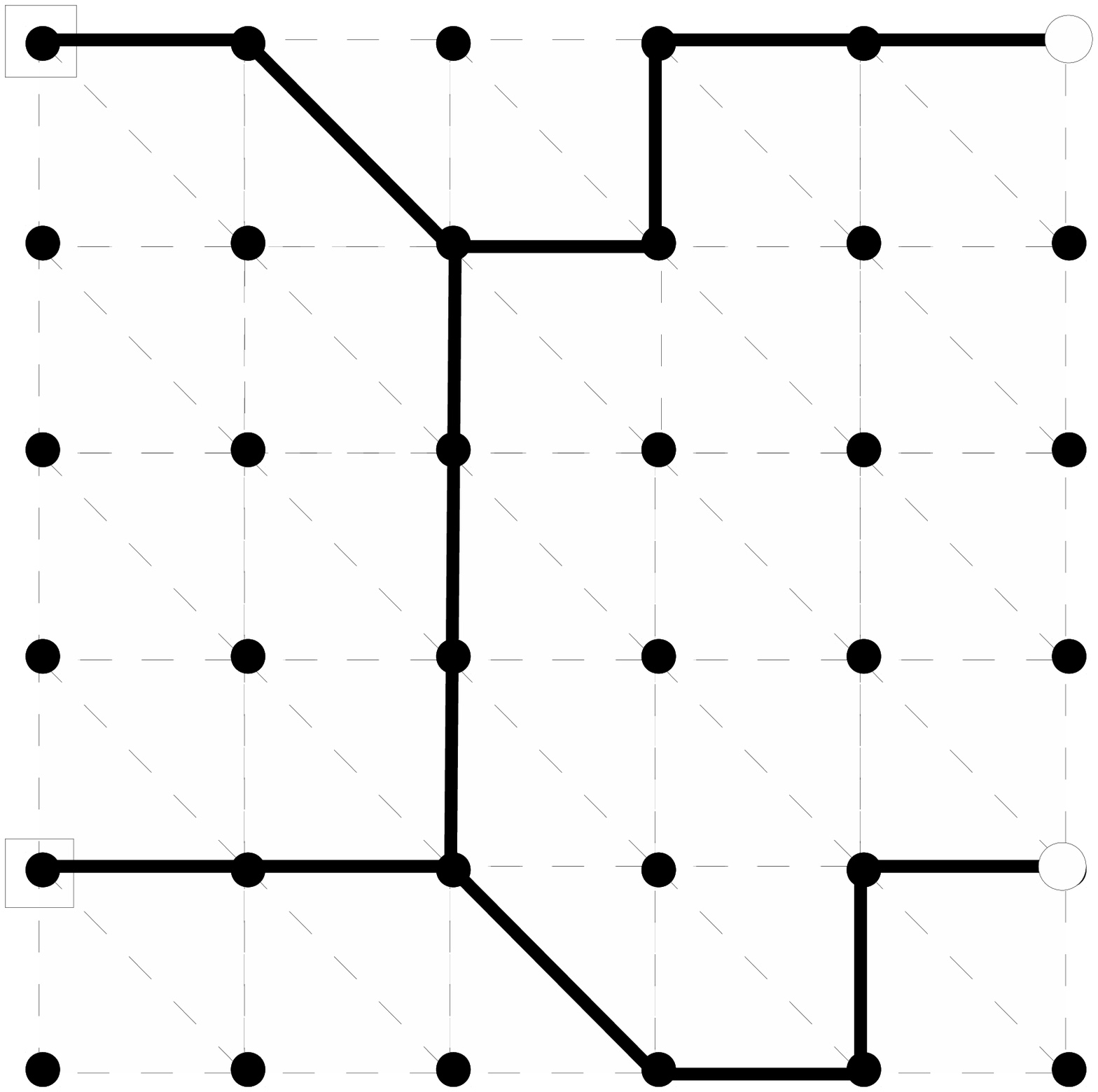} 
		\end{center}
\end{itemize}

As a consequence any graph is a pivot minor of a triangular grid of size at most $O(n^4)$ (=$O(n* d)= O(n*n^2*n)$)
\end{proof}

Other regular graph states, like rectangular\footnote{Graph states described by  rectangular grids are also called \emph{cluster states}.} or hexagonal grids, are commonly used in quantum computation. But, contrary to triangular grids, we show that some graphs are  pivot minor of no rectangular nor hexagonal grid:

\begin{lemma}\label{thm:rect}
There exists a graph which is a pivot minor of no rectangular nor hexagonal grid.
\end{lemma}
\begin{proof}
Pivot minors of bipartite graphs are bipartite \cite{Bouchet94-circle},
 thus, since any rectangular or hexagonal grid is bipartite, all pivot minors of these kind of grids are bipartite. As a consequence, any graph which is not bipartite (e.g. a triangle) is a pivot minor of no rectangular nor hexagonal grid.
\end{proof}
Theorem \ref{thm:rect} implies that some graph states, like the triangle, cannot be reached from the traditional cluster state (graph state described by a rectangular grid) by a sequence of $X$- and $Z$-measurements. Triangular grids turn out to be more adapted to one-way quantum computation than rectangular or hexagonal ones, as it is underlined in the next section, by proving that measurements  in the $(X,Z)$-plane on a triangular graph state is a universal model of quantum computation.

\section{Universality of $(X,Z)$-measurements}\label{sec:univ}

Briegel and Raussendorf \cite{RB01} have proved that any unitary transformation can be simulated by applying on  a rectangular grid, a sequence of one-qubit measurements, such that each measurement is either according to $Z$, or in the $(X,Y)$-plane, i.e. according to an observable $\cos(\alpha)X+\sin (\alpha)Y$ for some $\alpha$. We prove that any (real) unitary transformation can be simulated by a sequences of measurements in the  $(X,Z)$-plane, on triangular grids. 

\begin{theorem}
$(X,Z)$-measurements are universal for one-way quantum computation.
\end{theorem}

\begin{proof}
The set of unitary transformations $\{H, \Lambda Z, P(\alpha), \alpha \in [0,2\pi)\}$ where $P(\alpha)= \cos(\alpha/2) X+ \sin(\alpha/2) Z$ is universal for (real) quantum computation \cite{shi02}.  $\Lambda Z$ and $H$ can be implemented in the one-way model using $X$-measurements  as follows (see lemmas \ref{lem:H} and \ref{lem:CZ}):

\psfrag{X}[c][][1]{$X$}
		\centerline{	\includegraphics[scale=0.2]{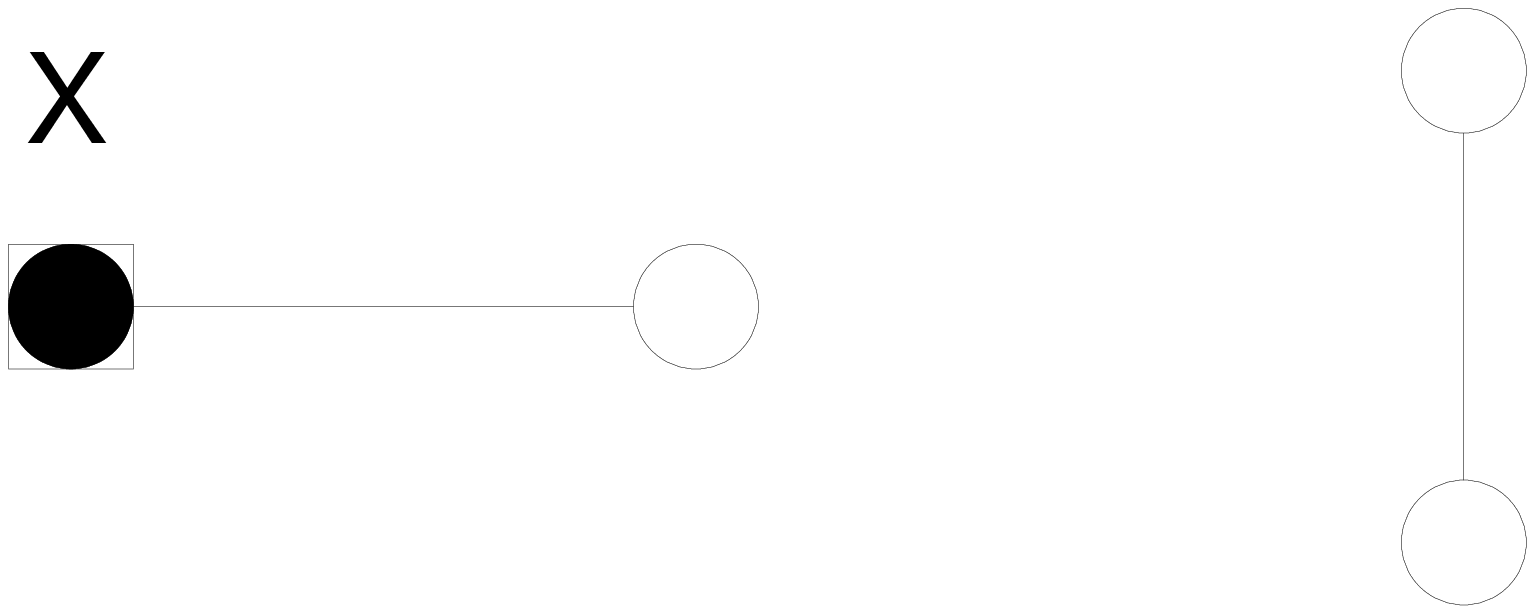} }
		%\end{center}
\vspace{-0.9cm}

Regarding $P(\alpha)$, it is implemented by the following one-way quantum computation using $(X,Y)$-measurements \cite{RBB03}:

\psfrag{X}[c][][1]{$X$}
\psfrag{Y}[c][][1]{$Y$}
\psfrag{planxy}[c][][1]{$\cos(\alpha)X+\sin(\alpha)Y$}
		\centerline{	\includegraphics[scale=0.2]{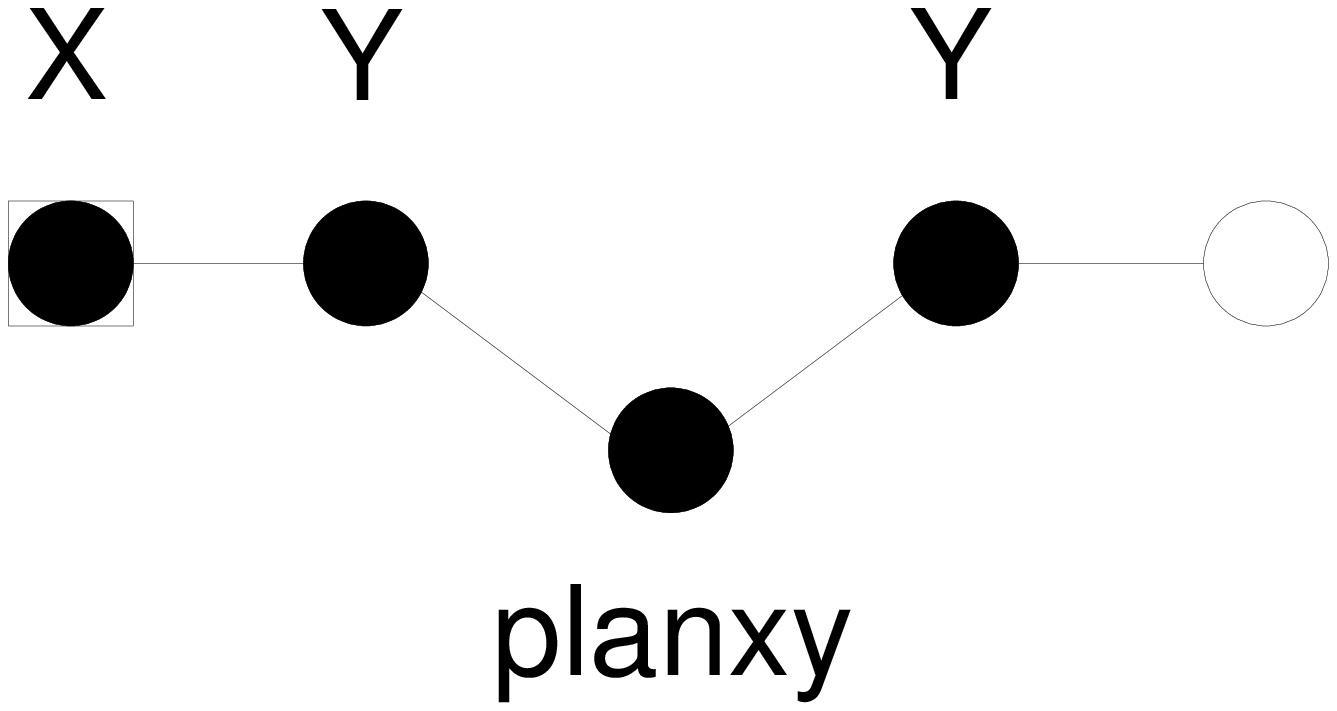} }
		%\end{center}
\vspace{-2.5cm}

This computation is equivalent to a one-way quantum computation involving measurement in the $(X,Z)$ plane only:

\psfrag{X}[c][][1]{$X$}
\psfrag{planxz}[c][][1]{$\cos(-\alpha)X+\sin(-\alpha)Z$}
		\centerline{	\includegraphics[scale=0.2]{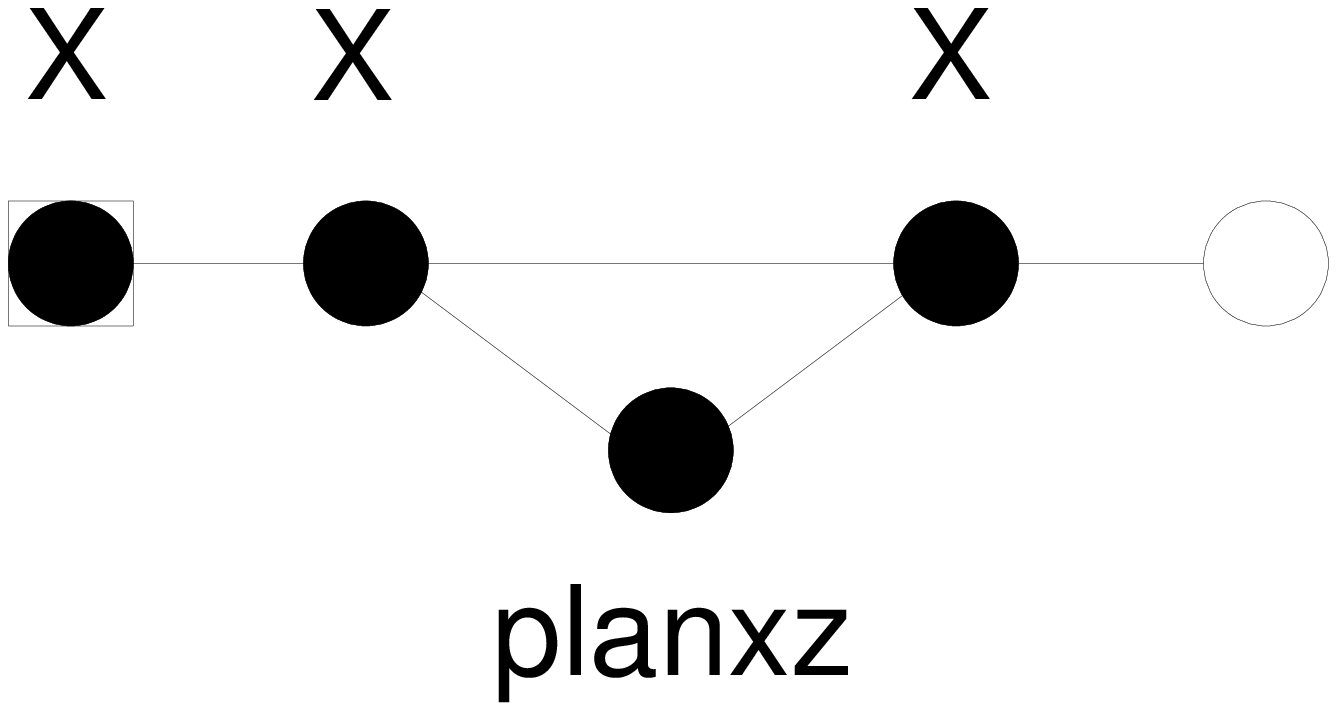} }
		%\end{center}
\vspace{-2.5cm}

The equivalency between the two computations can be proved using the local complementation. Indeed, the two underlying graphs are locally equivalent by applying a local complementation on the central vertex. Thus, if one takes the latter graph (with a triangle), applies $\sqrt X^\dagger$ on the central vertex and $\sqrt Z$ on its neighbors and then measures the central vertex according to $\cos(\alpha)X+\sin(\alpha)Y$ and its neighbors according to $Y$ then it leads to the simulation of $P(\alpha)$. A unitary  $U$ followed by a measurement according to an observable $A$ is equivalent to a measurement according to the rotated observable $U A U^\dagger$, thus the application of a measurement according to $\sqrt X^\dagger ( \cos(\alpha)X+\sin(\alpha)Y) \sqrt X = \cos(-\alpha)X+\sin(-\alpha)Z$ on the central vertex and according to $\sqrt Z Y \sqrt Z^\dagger  = X$ on its neighbors lead to the simulation of $P(\alpha)$.   \end{proof}

\begin{theorem}
Triangular graph states are universal resources for one-way quantum computation based on $(X,Z)$-measurements.
\end{theorem} 
\begin{proof}
Any real unitary can be implemented by composing the one-way computations using $(X,Z)$ measurements only (lemma \ref{lem:univ}). The graph state corresponding to this composition can be obtained from the triangular grid by means of $X$ and $Z$ measurements only (theorem \ref{thm:pivotgrid}).
\end{proof}

\end{document}